\documentclass[twocolumn,conference]{IEEEtran}

\usepackage{graphics}
\usepackage{epstopdf}
\usepackage[pdftex]{graphicx}
\usepackage{stfloats}
\usepackage{array}
\usepackage{float}
\usepackage[cmex10]{amsmath}
\usepackage{amsmath,amsfonts,amssymb}
\usepackage{graphicx,graphics}
\usepackage{enumerate}
\usepackage{color}
\usepackage{verbatim}
\usepackage{amsthm}
\usepackage{textcomp}
\usepackage{multirow}
\usepackage{subcaption}
\usepackage{siunitx}
\usepackage{amsmath,amsthm}
\usepackage{cleveref}
\usepackage{cite}
\usepackage[font=footnotesize]{caption}
\usepackage[english]{babel}
\addto\captionsenglish{}

\newtheorem{theorem}{Theorem}

\begin{document}
\graphicspath{{./Figures/}}
\title{On the Regimes in Millimeter wave Networks: Noise-limited or Interference-limited? }

\author{\IEEEauthorblockN{Solmaz Niknam$^\ast$ and Balasubramaniam Natarajan$^\dag$\\ }
\IEEEauthorblockA{$^\ast$$^\dag$Department of Electrical and Computer Engineering, Kansas State University, Kansas, USA\\
Email: \{$^\ast$slmzniknam, $^\dag$bala\}@ksu.edu\\}
\vspace{-0.5cm}
}
\maketitle
\begin{abstract}
Given the overcrowding in the 300~MHz--3~GHz spectrum, millimeter wave (mmWave) spectrum is a promising candidate for the future generations of wireless networks. With the unique propagation characteristics at mmWave frequencies, one of the fundamental questions to address is whether mmWave networks are noise or interference-limited. The regime in which the network operates significantly impacts the MAC layer design, resource allocation procedure and also interference management techniques. In this paper, we first derive the statistical characteristic of the cumulative interference in finite-sized mmWave networks considering configuration randomness across spatial and spectral domains while including the effect of blockages. Subsequently, using the derived interference model we set up a likelihood ratio test (LRT) (that is dependent on various network parameters) in order to detect the regime of the network from an arbitrarily located user standpoint. Unlike traditional networks, in mmWave networks, different likelihood of experiencing an interference-limited regime can be observed at different locations.
\end{abstract}

\section{Introduction} \label{sec:intro}
The availability of a large portion of millimeter wave (mmWave) spectrum has given rise to the idea that utilizing this chunk of spectrum may become a viable option in the next generation of wireless networks, e.g., 5G~\cite{Rappaport2013WillWork}. However, due to its challenging propagation characteristics including severe pathloss and strong atmospheric absorption, mmWave spectrum has been underutilized in mobile communication. Thanks to large antenna arrays that coherently direct the beam energy, highly directional signaling can help overcome the adverse mmWave propagation characteristics. However, utilization of directional beams changes many aspects of the wireless system design~\cite{Andrew2014what}. In fact, directional links are susceptible to blockages and obstacles~\cite{MacCartneyJr2016}. Highly narrow beams, large available bandwidth and high signal attenuation in mmWave spectrum may lead us to the conclusion that mmWave network performance is limited only by thermal noise (noise-limited regime). However, depending on the density of APs, density of the obstacles, transmission probability, and operating beamwidth, mmWave network performance may degrade due to interference (interference-limited regime). Unlike traditional wireless networks, mmWave networks may transit from a noise-limited regime to an interference-limited regime or exhibit intermediate behavior in which both regimes can be observed~\cite{measurment2016Mattia}. The regime in which the network is operating highly affects the MAC layer design and resource allocation strategies~\cite{mmWave2016Ghadikolahi}. Moreover, determining the network regime is critical in terms of identifying the most appropriate interference coordination technique that is effective in an interference-limited regime. However, when the network is in a noise-limited regime, we may not need any interference management mechanism or only a simple one may suffice. Therefore, one of the fundamental questions of interest in mmWave dense networks is whether the performance is limited by the interference or just by thermal noise.
%While in traditional wireless networks, the major factors determining the network regime are cell distances and transmit power spectral density, in mmWave networks more factors contribute~\cite{measurment2016Mattia}.
%Such transitional behavior of mmWave networks, mainly resulted from the dense environment of APs with pencil-beam signals, then requires reconsidering the MAC layer design and physical layer procedures.

There have been a few prior efforts focused on determining network regimes. \cite{dense2007ebrahimi,Adhoc2007nihal} have proposed conditions under which the network is noise or interference-limited. However, the density of the interfering APs are assumed to be fixed which may not be suitable for 5G mmWave networks that may exhibit uncertain spatial configurations due to factors like unplanned user-installed APs~\cite{5621983} and sensitivity to obstacles. In~\cite{measurment2016Mattia}, the network regime is determined, modeling the transmitter location as a Poisson point process (PPP). However, mmWave specifications such as severe pathloss and beam sensitivity to small-sized obstacles are not taken into consideration. The fact that interference power can change due to the presence of obstacles~\cite{intfmodel2016solmaz} limits the applicability of~\cite{measurment2016Mattia}.  In~\cite{mmWave2016Ghadikolahi}, the transition probability from a noise-limited to an interference-limited regime is calculated in a PPP mmWave network with random blockages. However, considering the spatial locations of the interfering APs as a PPP is not an appropriate choice for modeling finite-sized networks with fixed number of APs where performance becomes location dependent, as shown in~\cite{Haenggi2012stochBook}. Moreover, the blockage model used in~\cite{mmWave2016Ghadikolahi} is based on the unrealistic assumption of having a complete link outage with only one obstacle. However, in many practical mmWave applications such as indoor mmWave environments, outdoor mmWave small cells where coverage range is limited or even cases where terminals are equipped with larger number of antennas with wider beamwidths, more than one obstacle is needed to impact the power level, causing link blockage~\cite{Niu2015survey5G}.
%However, measurement has shown that even for a given density of the interfering APs, different regimes can be observed from different {users\textquotesingle} standpoint~\cite{measurment2016Mattia}. In addition, \cite{mmWave2016Ghadikolahi} has assumed a PPP process for the distribution of the interfering nodes which is not capable of properly reflecting the location-based performance of the networks with finite number of nodes in a limited area.

In this paper, we take a systematic approach to determine the network regime in mmWave networks. In order to overcome the limitation of prior efforts, we consider a more realistic and appropriate network and blockage models upon which the regime identification is formulated as a hypothesis testing problem. Specifically, we detect whether an arbitrarily located user experiences a noise or interference-limited regime based on the received signal power distribution in the presence of arbitrary-sized blockages. We calculate the distributions of the signal-plus-noise and signal-plus-interference powers which serve as the null and alternate hypotheses, respectively. In order to calculate the interference power, a 2D Binomial point process (BPP)~\cite{Solmaz2017Finite} is assumed to account for the randomness of interfering APs configuration in both spatial and spectral domains in a finite area\footnote{BPP is an appropriate choice in order to model the node locations in finite-sized networks with a given number of nodes~\cite{haenggi2012stochastic}.}. In fact, we consider a grid structure of space-frequency locations where interfering APs are placed randomly based on a BPP. We also account for beam directionality by including the effect of presence of arbitrary-sized blockages in the environment using a more realistic blockage model. It is notable that, unlike~\cite{mmWave2016Ghadikolahi} and other works on blockage modeling~\cite{Bai2014Blockage,gupta2017macro,muller2017analyzing,venugopal2016millimeter,thornburg2016performance}, in this blockage model the net effect of partial blockage caused by each individual obstacle is calculated. Since, more than only one obstacle may cause complete link blockage. Moreover, due to the inherent complexity in evaluating the exact distribution, an approximation of the distribution under the alternate hypothesis (signal-plus-interference power) is calculated using the maximum entropy (ME) technique~\footnote{Based on the principle of ME~\cite{cover2012elements}, ME distribution is the least informative distribution subject to specified properties or measures. Intuitively speaking, it has the minimum amount of prior information built into the distribution.}. Subsequently, using the standard likelihood ratio test (LRT) based on a Neyman-Pearson (NP) framework, we determine the regime of the network. It is important to note that determining the regime of the network is highly impacted by the interference model that appropriately reflects the network specifications. Therefore, the purpose of this paper is to leverage the detailed statistical interference model and its relation to various key deployment parameters including access point density, blockage density, transmit power, bandwidth and antenna pattern to provide an accurate assessment of the regime of the mmWave networks. It is also shown that the likelihood of experiencing an interference-limited regime depends on the interferer and blockage densities and varies at different spatial locations.

%The remainder of this paper is organized as follows: Section~\ref{sec:sys_model} describes the system model. In Section~\ref{sec:Regime_class}, we calculate the distribution of the received signal power under both null and alternate hypotheses (exact and approximation). Subsequently, considering NP criterion, the decision i.e., the LRT test is derived. Finally, the numerical results and conclusions are presented in Section~\ref{sec:Simulation} and Section~\ref{sec:Conclusion}, respectively.

\section{System Model} \label{sec:sys_model}
We consider a circular area of radius $R$ in 2D plane (${\rm I\!R}^2$) centered at the origin, with $N$ number of interfering APs operating in frequency band $[f_s,f_e]$. We also assume that a reference receiver, located at an arbitrary location ${v_0} \in B(O;R) = \left\{ {\left. x \in {\rm I\!R}^2 \right|\,\,{\left\| x \right\|_2} < R} \right\}$ with arbitrary frequency ${f_0} \in [f_s,f_e]$, is communicating with a reference transmitter over an intended communication link. This assumption gives the freedom of evaluating the network regime for users at different locations enabling more efficient resource management (e.g., interference coordination/cancellation only for those users whose performances are limited by interference). Interfering APs are distributed based on BPP in the space-frequency domain with success probability $p$. In other words, we consider a grid structure where the total $N$ interferers are randomly located at space-frequency locations based on a BPP\footnote{Reference transmitter-receiver pair is not a part of the point process.}. The overall received interference signal is the sum of the received signal from each interferer at a random space-frequency location. We also assume a random number of arbitrary-sized blockages in the environment distributed based on a PPP~\cite{venugopal2016millimeter,thornburg2016performance,muller2017analyzing,Bai2014Blockage} with parameter $\rho$. Due to the presence of the arbitrary blockages in the environment, the transmitted signal of interfering APs may be blocked and not all of the interfering APs contribute to the total received interference signal. Therefore, we are primarily concerned with the interferers that are in the line-of-sight (LoS) of the reference receiver.
\begin{figure}[t]
\centering
\includegraphics[scale=0.48]{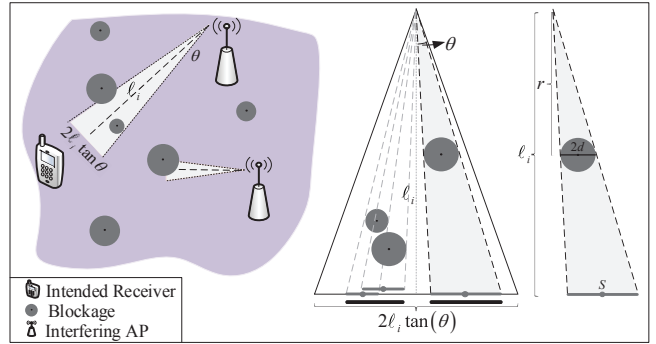}
\vspace{-1cm}
\caption{Radiation cone and the effective shadow of the blockages on the base of the radiation cone.}
\vspace{-0.5cm}
\label{fig:Shadow}
\end{figure}

In order to calculate the distribution of the number of active (non-blocked) interfering APs, we consider the blockage model presented in our prior work~\cite{Solmaz2017Finite}. In~\cite{Solmaz2017Finite}, the blockage effect is modeled by considering the net effect of partial blockage that each individual obstacle causes by intersecting the {interferers\textquotesingle} beam. In this model, obstacles are assumed to be modeled as circles with uniformly distributed radius $d$ in $[d_s,d_e]$. Assuming a radiation cone (see Fig.~\ref{fig:Shadow}) for the ${i^{{\rm{th}}}} \in \left\{ {1,2,...,N} \right\}$ interfering AP (where the edges are determined by the beamwidth of the signal, $2\theta$) we show that the average probability of each interfering AP being blocked corresponds to
\begin{align} \label{eq:blockage_prob}
{p_{\rm{b}}} =\frac{1}{{\frac{{{\mathbb E}\left[ d \right]}}{{2\tan \left( \theta  \right)}}}}{p_{{\rm{b1}}}} + \frac{1}{{{\mathbb E}\left[ \ell  \right] - \frac{{{\mathbb E}\left[ d \right]}}{{2\tan \left( \theta  \right)}}}}{p_{{\rm{b2}}}}.
\end{align}
Here, $p_{{\rm{b1}}}$ and $p_{{\rm{b2}}}$ are obtained using
\begin{align}\label{eq:pb1} \notag
&{p_{{\rm{b}}1}}=1 - \frac{{\sqrt {\frac{{\pi \tan (\theta )}}{\rho }} }}{{\left( {{d_e} - {d_s}} \right)}}\Bigg[ {\rm{erf}}( {{d_e}\sqrt {\frac{\rho }{{4\tan (\theta )}}} } ) \\
&\hspace{4cm} - {\rm{erf}}( {{d_s}\sqrt {\frac{\rho }{{4\tan (\theta )}}} })  \Bigg],
\end{align}
and
\begin{align} \label{eq:pb2}
{p_{{\rm{b}}2}} = \frac{{{{(1 + \Delta )}^{\left\lceil {\frac{\Delta }{{{{\rm{e}}^{\rho {\mathbb E}\left[ S \right]}} - 1}}} \right\rceil }}{{\rm{e}}^{ - (1 + \Delta )}}}}{{\left\lceil {\frac{\Delta }{{{{\rm{e}}^{\rho {\mathbb E}\left[ S \right]}} - 1}}} \right\rceil !}},
\end{align}
where $\Delta  = 2\rho {\mathbb{E}}\left[ \ell  \right]\tan \left( \theta  \right)$. Here, ${\mathbb{E}}[\ell]$ and ${\mathbb{E}}[S]$ denote the average distance from the interfering APs to the reference receiver and the average partial blockage caused by individual interfering APs, respectively. Given the BPP assumption of interfering nodes, the distribution of $\ell$ corresponds to
\begin{align} \label{eq:distance_distribution}
{f_L}\left( \ell  \right) {=} \left\{ \begin{array}{l}
\hspace{-0.15cm}\frac{{2\ell }}{{{R^2}}}\,\,\,\,\,\,\,\,\,\,\,\,\,\,\,\,\,\,\,\,\,\,\,\,\,\,\,\,\,\,\,\,\,\,\,\,\,\,\,\,\,\,\,\,\,\,\,\,\,\,\,\,\,\,\,\,\,\,\,\,\,\,\,\small{0 < \ell  \le R - \left\| {{v_0}} \right\|}\\
\hspace{-0.15cm}\frac{{2\ell {{\cos }^{ - 1}}\left( {\frac{{{{\left\| {{v_0}} \right\|}^2} - {R^2} +{\ell ^2}}}{{2\ell \left\| {{v_0}} \right\|}}} \right)}}{{\pi {R^2}}}\,\,\,\,\tiny{R - \left\| {{v_0}} \right\| < \ell  \le R + \left\| {{v_0}} \right\|}.
\end{array} \right.
\end{align}
In addition, the average partial blockage can be expressed as
\begin{align} \label{eq:Ave_S} \notag
{{\mathbb E}}\left[ S \right] &= {{\mathbb E}}\left[ {\frac{{2d\ell }}{r}\left| {d,r,\ell } \right.} \right] \\
&= \int\limits_{{d_s}}^{{d_e}} {\int\limits_{\frac{d}{{2\tan (\theta )}}}^{R + \left\| {{v_0}} \right\|} {\int\limits_{\frac{d}{{2\tan (\theta )}}}^\ell  {\frac{{2d\ell }}{r}{\rm{ }}{f_D}\left( d \right)f\left( {r,\ell } \right){\rm{d}}d\,{\rm{d}}r\,{\rm{d}}\ell } } }.
\end{align}
Detailed derivation of the blockage model is provided in our prior work~\cite{Solmaz2017Finite}.
\newtheorem{lemma}{Lemma}

Given the blockage probability in \eqref{eq:blockage_prob}, the distribution of the total number of non-blocked interfering APs is calculated using the following lemma:
\begin{lemma}\label{lem:non-blocked}
The total number of non-blocked interfering APs, denoted as $K$, is a Binomial random variable with success probability $p\left( {1 - {p_{\rm{b}}}} \right)$.
\end{lemma}
\begin{proof}
Let $K = {K_1} + {K_2} + ... + {K_N}$, where $K_i$ is a Bernoulli random variable and equals $1$, if the $i^{\rm {th}}$ interfering AP is not blocked, and $0$, otherwise. Therefore, the probability generating function (PGF) of $K$ is given by
\begin{equation}\label{eq:PGF}
{G_K} = (1 - {p_{\rm{b}}}){\rm{z}} + {p_{\rm{b}}}.
\end{equation}
Subsequently, we have
\begin{align}\label{eq:PGF_nonBlocked}\notag
{G_K}\left( {\rm{z}} \right)&= {{\mathbb E}}\Big[ {{{\mathop{\rm z}\nolimits} ^{\sum\limits_{i = 1}^N {{K_{i}}} }}} \Big]= \sum\limits_{k \ge 0} {{{\left( {{{\mathbb E}}\left[ {{{\mathop{\rm z}\nolimits} ^{K}}} \right]} \right)}^k}p\left( {N = k} \right)}\\ \notag
&= {G_N}\left( {{G_{K_i}}({\mathop{\rm z}\nolimits} )} \right)= {\left[ {\left( {1 - p} \right) + p\left( {(1 - {p_{\rm{b}}}){\rm{z}} + {p_{\rm{b}}}} \right)} \right]^N} \\
&= {\left[ {1 - p\left( {1 - {p_{\rm{b}}}} \right) + p\left( {1 - {p_{\rm{b}}}} \right){\rm{z}}} \right]^N},
\end{align}
 which is the PGF of a Binomial random variable with success probability $p\left( {1 - {p_{\rm{b}}}} \right)$.
\end{proof}

Now, having the distribution of the number of active interferers, in lemma~\ref{lem:non-blocked}, we set up a hypothesis test in order to determine the regime of the network.

\section{Regime Classification} \label{sec:Regime_class}
In this section, we formulate a binary hypothesis test where regime detection decision is based on the received power at an arbitrary located receiver. This hypothesis test is formally defined as:
\begin{align} \label{eq:Hyp}
\begin{array}{l}
{H_0}:\,\,\,\,Y = {\varphi} + \cal N \quad \quad \text{(Noise-limited regime)}\\
{H_1}:\,\,\,\,Y = {\varphi} + I \quad \quad \text{(Interference-limited regime),}
\end{array}
\end{align}
where ${\varphi}$, $I$ and $\cal N$ denote the average received power of the desired signal, aggregated interference and noise powers, respectively. We assume that the signal power is known and noise is characterized by a Gaussian random variable with mean $0$ and variance ${\sigma_n}^2$. Here, under $H_0$ hypothesis, the reference receiver experiences a noise-limited environment. This case may happen when most of the interfering APs are blocked by the blockages in the environment and the received interference power is low enough that the thermal noise is dominant. Alternately, under hypothesis $H_1$, the performance is limited by the received interference power. The distributions of the received power under both hypotheses need to be identified in order to derive the test.
%Subsequently, using the LRT method, we determine whether the network is in the interference-limited or noise-limited regime, in case of a specific network deployment scenario.
\subsection{Distribution under $H_0$} \label{subsec:Null_dist}
In this subsection, given the average received power of the desired signal, ${\varphi}$, we calculate the probability density function (PDF) of the received power under the null hypothesis.
\begin{lemma} \label{lem:H0}
The statistical distribution of the received power under $H_0$ is
\begin{align} \label{eq:H0}
{H_0}:\,\,\,\,Y \sim \frac{{{{\rm{e}}^{ - \,\,\frac{{y - {{\varphi}}}}{{2{\sigma ^2}}}}}}}{{2{\sigma ^2}\Gamma \left( {\frac{1}{2}} \right)\sqrt {\frac{{y - {{\varphi}}}}{{2{\sigma ^2}}}} }}
\end{align}
\end{lemma}
\begin{proof}
Since noise is assumed to be Gaussian with mean $0$ and variance ${\sigma_n}^2$, the distribution of the noise power is ${\cal N} \sim {\sigma_n}^2\chi _1^2$, where $\chi _1^2$ denotes a chi-squared distribution with $1$ degree of freedom. Consequently, the cumulative density function (CDF) of the power is given by
\begin{align}
{F_Y}\left( {Y \le y} \right) = {F_Y}\left( {{\cal N} \le y - {{\varphi}}} \right) = \,\frac{1}{{\Gamma \left( {\frac{1}{2}} \right)}}\gamma \left( {\frac{1}{2},\frac{{y - {{\varphi}}}}{{2{\sigma ^2}}}} \right).
\end{align}
Then, by taking derivative of the CDF, the PDF of the power under the null hypothesis is determined as
\begin{align} \notag
{\left. {{f_Y}(y)} \right|{H_0}} &= \frac{{\rm{d}}}{{{\rm{d}}y}}\left\{ {\frac{1}{{\Gamma \left( {\frac{1}{2}} \right)}}\gamma \left( {\frac{1}{2},\frac{{y - {{\varphi}}}}{{2{\sigma ^2}}}} \right)} \right\} \\ \notag
& = \frac{{\rm{d}}}{{{\rm{d}}y}}\left\{ {\frac{1}{{\Gamma \left( {\frac{1}{2}} \right)}}\int\limits_0^{\frac{{y - {{\varphi}}}}{{2{\sigma ^2}}}} {{t^{\frac{1}{2} - 1}}{{\rm{e}}^{ - t}}{\rm{d}}t} } \right\} \\
&= \frac{{{{\rm{e}}^{ - \,\,\frac{{y - {{\varphi}}}}{{2{\sigma ^2}}}}}}}{{2{\sigma ^2}\Gamma \left( {\frac{1}{2}} \right)\sqrt {\frac{{y - {{\varphi}}}}{{2{\sigma ^2}}}} }}.
\end{align}
\end{proof}

%Since, noise is Gaussian, the statistical model for $H_0$ is a non-central chi-squared distribution. Once $H_1$ can be modeled statistically, as well, the likelihood ratio test (LRT) can be formulated. In the next subsection, we derive the statistical model of $H_1$ which is based on the various key deployment parameters including AP density, blockage density, transmit power, bandwidth and antenna pattern. Moreover, by changing the deployment parameters, the regime of the network is determined from an arbitrary located user's point of view.
\subsection{Distribution under $H_1$} \label{subsec:Alt_dist}
In this subsection, we calculate the distribution of the received power in the interference-limited regime where the power of the noise is negligible. Therefore, the power received by an arbitrarily-located reference receiver is
\begin{align} \label{eq:received_power}
Y={\varphi}+{\sum\limits_{i = 1}^K {{{\cal P}_{{I_i}}}} },
\end{align}
where $\mathcal{P}_{I_i}$ is the effective received interference power from the $i^\text{th}$ interfering AP at the output of the matched filter which corresponds to~\cite{Hamdi2009unified},
\begin{align} \label{eq:Intf_pow}
{{{\cal P}_{{I_i}}}}={q_i}{h_i}{\left\| {\ell_i} \right\|^{ - \alpha }}\Upsilon \left( {\omega_i} \right).
\end{align}
Here, $h_i$ and ${\left\| . \right\|^{ - \alpha }}$ model the Nakagami-$m$ small scale fading and pathloss effects, respectively. $\ell_i={v_0} - {v_i}$ and $\omega_i={f_i}-{f_0}$ denote the spatial and spectral distance between the $i^{\rm{th}}$ interfering AP and the reference receiver, respectively. $q_i$ is the transmitted power of the $i^{\rm{th}}$ interfering AP.  Moreover, $\Upsilon \left( {{\omega_i}} \right)$ is defined as
\begin{align}
\Upsilon(\omega_i)=\int\limits_{{f_0} - \frac{W}{2}}^{{f_0} + \frac{W}{2}} {\Phi \left( {f - {f_i}} \right){{\left| {H\left( {f - {f_0}} \right)} \right|}^2}{\mathop{\rm d}\nolimits} f} ,
\end{align}
where $H(f-f_0)$ is the transfer function of the matched filter at the reference receiver with arbitrary frequency $f_0$, and $\Phi(f-f_i)$ is the power spectral density of the baseband equivalent of the interferers signals. Considering~\eqref{eq:Intf_pow}, as ${\left\| {\ell_i} \right\|^{ - \alpha }}$ captures the impact of spatial distances (and thereby random spatial configuration), $\Upsilon(\omega_i)$ accounts for the effect of frequency separation (and thereby random spectral configuration) in the interference power. The statistical distribution of the received signal power in the alternate hypothesis $H_1$, in terms of MGF, is obtained using the following theorem:
\begin{theorem} \label{theo:theo1}
The moment generating function (MGF) of $Y$, under alternate hypothesis $H_1$, is given by
\begin{align} \label{eq:MGF_alt_theo}
{M_{{Y}}}({\rm{s}}) ={{\rm{e}}^{{{\varphi}}{\rm{s}}}}{\left[ {1 - p\left( {1 - {p_{\rm{b}}}} \right) + p\left( {1 - {p_{\rm{b}}}} \right){M_{{{{\cal P}_{{I_i}}}}}\left( \rm{s} \right)}} \right]^N},
\end{align}
where,
\begin{align} \label{eq:MGF_indivi}
&{M_{{{\cal P}_{{I_i}}}}}\left( {\rm{s}} \right){=} \sum\limits_{n = 0}^\infty  {\frac{{{(q\, {\rm s})^n}}}{{n!}}\frac{{{m^{ - n}}\Gamma \left( {n + m} \right)}}{{\Gamma \left( m \right)}}\frac{{2\gamma_n \left( {{f_s},{f_e}} \right)\kappa_n \left( {R,{v_0}} \right)}}{{{R^2}\left( {{f_e} - {f_s}} \right)}}}  .\\ \notag
\end{align}
Here, ${M_{{{\cal P}_{{I_i}}}}}\left( {\rm{s}} \right)$ is the MGF of the $i^ {\rm{th}}$ interferer's power and
\begin{align}
{\gamma_n \left( {{f_s},{f_e}} \right)}=\hspace{-0.3cm}{\int\limits_0^{\min \left( {\left| {{\omega _e}} \right|,\left| {{\omega _s}} \right|} \right)} \hspace{-0.3cm}{\Upsilon {{\left( \omega  \right)}^n}{\rm{d}}\omega }  +\hspace{-0.3cm} \int\limits_0^{\max \left( {\left| {{\omega _e}} \right|,\left| {{\omega _s}} \right|} \right)} \hspace{-0.3cm}{\Upsilon {{\left( \omega  \right)}^n}{\rm{d}}\omega } },
\end{align}
and
\begin{align} \notag
{\kappa_n \left( {R,{v_0}} \right)}&=\hspace{-0.1cm}\int\limits_0^{R - \left\| {{v_0}} \right\|} {\hspace{-0.3cm}{\ell ^{ - n\alpha  + 1}}\,{\rm{d}}\ell \,\,} \\
&+\hspace{-0.1cm} \int\limits_{R - \left\| {{v_0}} \right\|}^{R + \left\| {{v_0}} \right\|} {\hspace{-0.2cm}\frac{{{\ell ^{ - n\alpha  + 1}}}}{\pi }{{\cos }^{ - 1}}\left( {\frac{{{{\left\| {{v_0}} \right\|}^2} - {R^2} + {\ell ^2}}}{{2\ell \left\| {{v_0}} \right\|}}} \right)\,{\rm{d}}\ell } .
\end{align}
\end{theorem}
\begin{proof}
In order to calculate the MGF of the received signal power under alternate hypothesis, we have
\begin{align} \label{eq:MGF_alt_det} \notag
{M_{{Y}}}({\rm{s}}) &= {{\mathbb E}}\left[ {{{\rm{e}}^{{\rm{s}}\left( {{{\varphi}} + \sum\limits_{i = 1}^K {{{\cal P}_{{I_i}}}} } \right)}}} \right]{\rm{ = }}{{\mathbb E}}\left[ {{{\rm{e}}^{{\rm{s}}{{\varphi}} + {\rm{s}}\sum\limits_{i = 1}^K {{{\cal P}_{{I_i}}}} }}} \right] \\ \notag
&{\rm{ = }}{{\rm{e}}^{{{\varphi}}{\rm{s}}}}{{\mathbb E}}\left[ {{{\rm{e}}^{{\rm{s}}\sum\limits_{i = 1}^K {{{\cal P}_{{I_i}}}} }}} \right] = {{\rm{e}}^{{{\varphi}}{\rm{s}}}}\sum\limits_{k \ge 0} {{{\left( {{{\mathbb E}}\left[ {{{\rm{e}}^{{\rm{s}}{{\cal P}_{{I_i}}}}}} \right]} \right)}^k}p\left( {K = k} \right)}   \\ \notag
&= {{\rm{e}}^{{{\varphi}}{\rm{s}}}}{G_K}\left( {M_{{{{\cal P}_{{I_i}}}}}\left( \rm{s} \right)} \right)\\
&= {{\rm{e}}^{{{\varphi}}{\rm{s}}}}{\left[ {1 - p\left( {1 - {p_{\rm{b}}}} \right) + p\left( {1 - {p_{\rm{b}}}} \right){M_{{{{\cal P}_{{I_i}}}}}\left( \rm{s} \right)}} \right]^N}.
\end{align}
where ${M_{{{\cal P}_{{I_i}}}}}\left( {\rm{s}} \right)$ is the MGF of the $i^ {\rm{th}}$ interferer's power. The MGF of the power of the individual interfere is calculated as
\begin{align} \label{eq:MGF_indivi_det} \notag
{M_{{{\cal P}_{{I_i}}}}}\left( {\rm{s}} \right) &= {{\mathbb E}}\left[ {{{\rm{e}}^{{\rm{s}}\,q\,h\,{{ \ell  }^{ - \alpha }}\Upsilon \left( \omega  \right)}}} \right] \\ \notag
&= \int\limits_0^\infty  \int\limits_0^{R + \left\| {{v_0}} \right\|} \int\limits_0^{\max \left( {\left| {{\omega _e}} \right|,\left| {{\omega _s}} \right|} \right)} {{\rm{e}}^{{\rm{s}}\,q\,h\,{{ \ell  }^{ - \alpha }}\Upsilon \left( \omega  \right)}}\\
&\hspace{3 cm} \times \,\,{f_\Omega }\left( \omega  \right){f_L}\left( \ell  \right)f\left( h \right){\rm{d}}\omega \,{\rm{d}}\ell \,{\rm{d}}h.
\end{align}
Given the BPP assumption of the location of interferer in space-frequency domain, the distributions of spectral distance is given by~\cite{Solmaz2017Finite}
\begin{align}\label{eq:frequency_distribution}
&{f_\Omega }\left( \omega  \right) {=} \left\{ \begin{array}{l}
\frac{2}{{{f_e} - {f_s}}}\,\,\,\,\,\,\,\,\,\,\,\,\,\,0 < \omega  \le \min \left( {\left| {\omega_{e}} \right|,\left| {\omega_{s}} \right|} \right)\\
\frac{1}{{{f_e} - {f_s}}}\,\,\,\,\min \left( {\left| {\omega_{e}} \right|,\left| {\omega_{s}} \right|} \right) < \omega  \le \max \left( {\left| {\omega_{s}} \right|,\left| {\omega_{s}} \right|} \right),
\end{array} \right.
\end{align}
where $\omega_{e}={f_e} - {f_0}$ and $\omega_{s}={f_s} - {f_0}$. Having the spatial and spectral distance distributions given in~\eqref{eq:distance_distribution} and~\eqref{eq:frequency_distribution}, Nakagami-$m$ assumption of the small scale fading, i.e., $h$ and using the polynomial expansion of the exponential function, the integral in~\eqref{eq:MGF_alt_det} is derived as in~\eqref{eq:MGF_indivi}. Subsequently, by substituting~\eqref{eq:MGF_indivi} in~\eqref{eq:MGF_alt_det}, the result in~\eqref{eq:MGF_alt_theo} is obtained.
\end{proof}

Calculating the inverse Laplace transform of the MGF in~\eqref{eq:MGF_alt_theo} to find the distribution is a tedious task and computationally complex. A more straightforward method might be to approximate the given distribution with known distributions. Therefore, in order to make the problem tractable, we use the ME method. Specifically, the ME technique~\cite{cover2012elements} is used to approximate the distribution of the received power under the alternate hypothesis with a tractable and simpler form. Basically, ME estimate is an estimate with maximal information entropy (least-informative) subject to the given moments. All the information about the interference power distribution is provided by the MGF in~\eqref{eq:MGF_alt_theo}. Therefore, we can use as many moments as needed (as the prior information or constraints) to make the estimation more precise. However, considering higher number of moments leads to the calculation of sets of non-linear equations which itself adds complexity to the problem. Here, for simplicity, we use the first moment of the received power (mean received power) as the constraint while maximizing the entropy of the distribution. We believe that this is a reasonable and logical starting point as the difference in mean power of interference and noise should offer the greatest discriminatory effect between $H_0$ and $H_1$. Later, Section~\ref{sec:Simulation}, we show the performance of the resulting test using receiver operating characteristic (ROC) curve and it can been considered as the lower bound on the performance of the ideal test with the true distribution under $H_1$. In fact, the test performance can be improved by including the higher order moments as part of the ME estimation constraints at the cost of increasing model complexity. This will be explored as part of our future work.

\begin{lemma}\label{lem:H1}
The approximated PDF of the received signal power under alternate hypothesis $H_1$ is given by
\begin{align} \label{eq:H1_approx}
{H_1}:{\kern 1pt} {\kern 1pt} {\kern 1pt} {\kern 1pt} Y\sim \lambda {{\rm{e}}^{ - \lambda \left( {y - {{\varphi}}} \right)}},
\end{align}
where $\lambda$ is derived by solving $\left( {\lambda {{\varphi}} + 1} \right){{\rm{e}}^{ - \lambda {{\varphi}}}} - {\mathbb E}\left[ y \right]{\lambda ^2} = 0$ and ${\mathbb E}\left[ y \right]$ is the first moment of the received signal power under the alternate hypothesis.
\end{lemma}
\begin{proof}
Considering the first moment as the constraint in the ME method, the PDF of the received signal power can be calculated by solving
\begin{align} \label{eq:opt_problm}
\begin{array}{l}
\max \,\,\,\,\,\, - {f_Y}\left( y \right)\ln \left( {{f_Y}\left( y \right)} \right)\\
{\rm{s}}{\rm{.t}}{\rm{. }}\,\,\,\,\,\,\,\int\limits_{{{\varphi}}}^\infty  {y\,{f_Y}\left( y \right)}  = {\mathbb E}\left[ y \right].
\end{array}
\end{align}
Here, ${\mathbb E}\left[ y \right]={\left. {\frac{\partial }{{\partial {\rm{s}}}}{{\rm{M}}_Y}\left( {\rm{s}} \right)} \right|_{{\rm{s}} = 0}}$. The ME probability is found using the dual Lagrangian method~\cite{cover2012elements},
\begin{align} \label{eq:Lagrangian_derivative}
\frac{\partial }{{\partial {f_Y}\left( y \right)}}L\left( {{f_Y}\left( y \right),\lambda } \right) = 0,
\end{align}
where,
\begin{align} \notag
&L\left( {{f_Y}\left( y \right),\lambda } \right)=  - {f_Y}\left( y \right)\ln \left( {{f_Y}\left( y \right)} \right)\\
 &\hspace{2cm} + \lambda \left( {\int\limits_{{{\varphi}}}^\infty  {y\,{f_Y}\left( y \right)}  - {\mathbb E}\left[ y \right]} \right),
\end{align}
is the Lagrangian of the optimization problem~\eqref{eq:opt_problm}. By solving~\eqref{eq:Lagrangian_derivative} using the Karush-Kuhn-Tucker (KKT) conditions, the distribution ${f_Y}\left( y \right)$ is derived as in~\eqref{eq:H1_approx}, where $\lambda$ is calculated by solving $\left( {\lambda {{\varphi}} + 1} \right){{\rm{e}}^{ - \lambda {{\varphi}}}} - {\mathbb E}\left[ y \right]{\lambda ^2} = 0$.
\end{proof}
\subsection{Likelihood Ratio Test} \label{subsec:LRT}
With the knowledge of distribution of $Y$ under both $H_0$ and $H_1$, we can write down the likelihood ratio as
\begin{align} \label{eq:LRT} \notag
{\mathop{\rm LRT}\nolimits} (y) &= \frac{{\left. {{f_Y}(y)} \right|{H_1}}}{{\left. {{f_Y}(y)} \right|{H_0}}} = \frac{{\lambda {{\rm{e}}^{ - \lambda \left( {y - {{\varphi}}} \right)}}}}{{\frac{{{{\rm{e}}^{ - \,\,\frac{{y - {{\varphi}}}}{{2{\sigma ^2}}}}}}}{{2{\sigma ^2}\Gamma \left( {\frac{1}{2}} \right)\sqrt {\frac{{y - {{\varphi}}}}{{2{\sigma ^2}}}} }}}}\\
& = 2\Gamma \left( {\frac{1}{2}} \right){\sigma ^2}\lambda {{\rm{e}}^{\left( { - \lambda  + \frac{1}{{2{\sigma ^2}}}} \right)y + \lambda {{\varphi}} - \frac{{{{\varphi}}}}{{2{\sigma ^2}}}}}\sqrt {\frac{{y - {{\varphi}}}}{{2{\sigma ^2}}}}.
\end{align}
Considering the well-known \emph{Neyman-Pearson} (NP) criterion, the decision rule, i.e., ${\delta _{{\rm{NP}}}}$ is
\begin{align} \label{eq:test}
{\delta _{{\rm{NP}}}} = \left\{ \begin{array}{l}
1\,\,\,\,\,\,\,\,{\rm{LRT}}\left( y \right) \ge \eta  \Rightarrow \,\,y \ge {\rm{LR}}{{\rm{T}}^{ - 1}}(\eta ) = \eta'\\
0\,\,\,\,\,\,\,\,{\rm{LRT}}\left( y \right) < \eta  \Rightarrow \,\,y < {\rm{LR}}{{\rm{T}}^{ - 1}}(\eta ) = \eta'.
\end{array} \right.
\end{align}
It is notable that, the NP framework is chosen in order to prevent the imposition of a specific cost to the decision made and priors on the hypotheses. In order to calculate the threshold $\eta '$, we have
\begin{align} \label{eq:eta_prime} \notag
{P_F}\left( \delta_{NP}  \right) = {\beta _{{\rm{th}}}} &\Rightarrow \int\limits_{\eta '}^\infty  {\frac{{{{\rm{e}}^{ - \frac{{y - {{\varphi}}}}{{2{\sigma ^2}}}}}}}{{2{\sigma ^2}\Gamma \left( {\frac{1}{2}} \right)\sqrt {\frac{{y - {{\varphi}}}}{{2{\sigma ^2}}}} }}\,{\rm{d}}y}  = {\beta _{{\rm{th}}}} \\
& \Rightarrow \eta ' = 2{\sigma ^2}{\left( {{{{\mathop{\rm erf}\nolimits} }^{ - 1}}\left( {1 - {\beta _{{\rm{th}}}}} \right)} \right)^2} + {{\varphi}},
\end{align}
where $\beta_{\rm{th}}$ denotes the \emph{significance level} of the test. Having the threshold $\eta '$ in~\eqref{eq:eta_prime}, the detection probability is obtained as
\begin{align} \label{eq:detection_prob}
{P_D}\left( \delta_{NP}  \right) = \int\limits_{\eta '}^\infty  {\lambda {{\rm{e}}^{ - \lambda \left( {y - {{\varphi}}} \right)}}\,{\rm{d}}y}   = {e^{ - \lambda \left( {\eta ' - {{\varphi}}} \right)}}.
\end{align}

\section{Numerical Results} \label{sec:Simulation}
In this section, we present numerical results to determine the performance of the test given the various key deployment parameters. A circular area of radius $R=10$ m is considered. The reference receiver is located at spectral location $f_0=62$ GHz. Moreover, $f_s$ and $f_s$ are set to 58 GHz and 64 GHz, respectively. The pathloss exponent, $\alpha$, and the shape factor of Nakagami distribution, $m$, are set to 2.5 and 3, respectively. Here, the transmitted power of all interfering APs are assumed to be the same and set to 27 dBm. The beamwidth of the mmWave signals, $2\theta$, is set to 20 degrees. We assume Gaussian PSD for interfering APs and an RC-0 for the matched filter at the reference receiver side. It is worth mentioning that the proposed model and the hypothesis test are not limited to specific power spectral densities or pulse shape choices of the desired and {interferers\textquotesingle} signals.

In Fig.~\ref{fig:LRT_rho}, the area under the LRT curve is shown as a function of distance from the origin (for fixed number of interfering APs). Since, \eqref{eq:H0} is independent of the reference receiver's location; therefore, higher values in Fig.~\ref{fig:LRT_rho} represent the higher values in~\eqref{eq:H1_approx} which means the higher likelihood of being in the interference-limited regime. When the density of blockages increase, more interfering APs are blocked. Therefore, there is less number of interfering APs that introduce interference to the reference receiver and the probability of being in the interference regime decreases. In addition, it can be seen that the probability of experiencing the interference-limited regime decreases as the reference receiver moves from the center of the area to its periphery. The same trend can be observed as the number of interfering APs changes (with the fixed blockage density), as shown in Fig.~\ref{fig:LRT_N}. The scenarios in which the effect of the presence of blockages is not considered is also provided in Fig.~\ref{fig:LRT_N}. Here, we can see how ignoring the blockage effect results in an overestimation in the likelihood of observing an interference-limited regime.

In Fig.~\ref{fig:ROC_N}, ROC curve is shown for different number of interfering APs and blockage density $\rho$ set to $1$. Here, the detection probability represents the probability of detecting an interference-limited regime for a specific set of deployment parameters.

As we can see in the results, the derived distributions are functions of various key deployment parameters including access point density, blockage density, transmit power, bandwidth and antenna beamwidth. Using the binary hypothesis test in~\eqref{eq:test} we can decide, given a specific set of deployment parameters for the network, which regime is more probable for receivers located at different locations in the finite-sized network.

\begin{figure}[t]
\centering
\includegraphics[height=4.5cm,width=7cm]{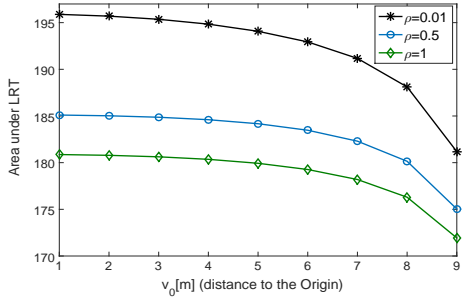}
\caption{LRT versus reference receiver location for different $\rho$, $N{=}200$.}
\label{fig:LRT_rho}
\vspace{-0.6cm}
\end{figure}
\begin{figure}[t]
\centering
\includegraphics[height=4.5cm,width=7cm]{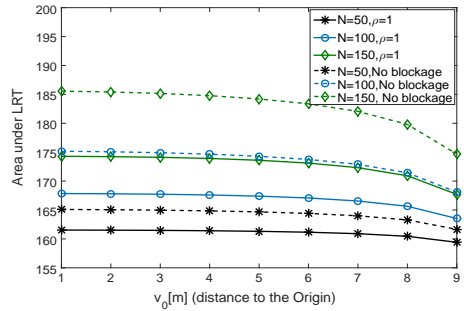}
\caption{LRT versus reference receiver location for different N values.}
\label{fig:LRT_N}
\vspace{-0.6cm}
\end{figure}
\begin{figure}[t]
\centering
\includegraphics[height=4.5cm,width=6.8cm]{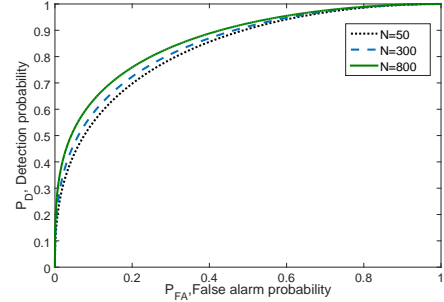}
\caption{ROC curve for different N values, $\rho{=}1$.}
\label{fig:ROC_N}
\vspace{-0.6cm}
\end{figure}
\vspace{-0.3cm}
\section{Conclusion} \label{sec:Conclusion}
In this paper, we set up a binary hypothesis test based on the received signal power in order to detect the regime of the mmWave networks in the presence of the blockages. We derive the distributions of the signal-plus-noise and signal-plus-interference powers, i.e., the power distributions in the case of null and alternate hypotheses of the binary test, respectively. Due to the complexity of the derived distribution under alternate hypothesis and in order to make the problem tractable, we leverage the method of maximum entropy to approximate the distribution. Using the approximated distribution and deploying the Neyman-Pearson criterion, we calculate the probability of experiencing an interference-limited regime. It is worth reiterating that the detailed statistical interference model and its relation to various key deployment parameters including access point density, blockage density, transmit power, bandwidth and antenna pattern helps provide an accurate assessment of the network regimes at different locations in the network.

\bibliographystyle{IEEEtran}

\bibliography{IEEEabrv,GBbibfile}

\end{document}